\documentclass{article}
\usepackage{geometry}
\usepackage[latin1]{inputenc}

\usepackage{amsmath,amsthm,amssymb,subdepth}
\usepackage{amsfonts,mathrsfs,enumerate}

\usepackage{hyperref}

\title{The Phase Space for the Einstein-Yang-Mills Equations and the First Law of Black Hole Thermodynamics}
\date{\today}
\author{Stephen\hspace{-1.5mm}\newlength{\Mheight}
\newlength{\cwidth}
\settoheight{\Mheight}{M}\settowidth{\cwidth}{c}M\parbox[b][\Mheight][t]{\cwidth}{c}\hspace{0.15mm}Cormick\footnote{stephen.mccormick@monash.edu}\vspace{3mm}\\School of Mathematical Sciences\\Monash University\\Clayton 3800\\Australia}
\newtheorem{theorem}{Theorem}[section]

\newtheorem{corollary}[theorem]{Corollary}

\newtheorem{lemma}[theorem]{Lemma}

\newtheorem{proposition}[theorem]{Proposition}

\newcommand{\onabla}{\mathring\nabla}
\newcommand{\lie}[1]{\mathfrak{#1}}

\numberwithin{equation}{section}
\DeclareMathOperator{\Ran}{Ran}

\DeclareMathOperator{\tr}{tr}

\DeclareMathOperator{\coker}{coker}

\begin{document}
\maketitle
\begin{abstract}
We use the techniques of Bartnik \cite{phasespace} to show that the space of solutions to the Einstein-Yang-Mills constraint equations on an asymptotically flat manifold with one end and zero boundary components, has a Hilbert manifold structure; the Einstein-Maxwell system can be considered as a special case. This is equivalent to the property of linearisation stability, which was studied in depth throughout the 70s \cite{armsEM,armsEYM,BDinstability,C-BDstability,FM1,moncrief1,moncrief2}.

This framework allows us to prove a conjecture of Sudarsky and Wald \cite{SW1}, namely that the validity of the first law of black hole thermodynamics is a suitable condition for stationarity. Since we work with a single end and no boundary conditions, this is equivalent to critical points of the ADM mass subject to variations fixing the Yang-Mills charge corresponding exactly to stationary solutions. The natural extension to this work is to prove the second conjecture from \cite{SW1}, which is the case where an interior boundary is present; this will be addressed in future work.

\end{abstract}
\section{Introduction}
A solution to the full Einstein-Yang-Mills equations is given by a Lorentzian metric and a $\lie{g}-$valued one-form on a 4-dimensional manifold $^4\hspace{-0.6mm}\mathcal{M}$, where $\lie{g}$ is the Lie algebra of some compact Lie group $\hat{G}$. It is well known, that given a sufficiently regular solution, $(g,A,\pi,\varepsilon)$ of the constraint equations (\ref{constraint1})-(\ref{gauss}) on a 3-manifold, $\mathcal{M}$, we can find a full solution to the Einstein-Yang-Mills equations with an embedded hypersurface on which the initial data is induced. By this, we mean $g$ is the induced metric; $\pi=(K-tr(K))\sqrt{g}$, where $K$ is the second fundamental form; $A$ is the projection of the Yang-Mills connection onto the hypersurface; and $\varepsilon$ is four times the negative of the induced Yang-Mills electric field density, $E$, as viewed by a Gaussian normal set of observers. We use $\pi$ and $\varepsilon$ instead of $K$ and $E$ as these quantities are the canonical momenta for the Hamiltonian formulation used.

Geometrically, the Yang-Mills fields can be thought of as coming from some principal $\hat{G}$-bundle,  $^4\hspace{-0.6mm} P$;  for simplicity, we assume this bundle is trivial, $^4\hspace{-0.6mm} P\cong {}^4\hspace{-0.6mm}\mathcal{M}\times \hat{G}$. The interpretation of the Yang-Mills potential, $^4\hspace{-0.6mm}A$ is as the pullback of a connection one-form $\omega$ on $^4\hspace{-0.6mm}P$ via a global section. Given a global section $\iota: {}^4\hspace{-0.6mm}\mathcal{M}\rightarrow {}^4\hspace{-0.6mm}\mathcal{M}\times \hat{G}$, we define $^4\hspace{-0.6mm}A=\iota^*\omega$ and $^4\hspace{-0.6mm}F=\iota^*(\Omega)$, where $\Omega=d\omega+\omega\wedge\omega$ is the curvature form of $\omega$. The Yang-Mills potential $^4\hspace{-0.6mm}A$ is a $\lie{g}$-valued one-form on $^4\hspace{-0.6mm}\mathcal{M}$ and $^4\hspace{-0.6mm}F$ is a $\lie{g}$-valued two-form on $^4\hspace{-0.6mm}\mathcal{M}$, called the field strength tensor. The restriction of $^4\hspace{-0.6mm}A$ and $^4\hspace{-0.6mm}F$ to $\mathcal{M}$ can be viewed respectively, as the pullback via some section of a connection one-form and associated curvature on a bundle $P\cong \mathcal{M}\times \hat{G}$, the restriction of $^4\hspace{-0.6mm}P$ to $\mathcal{M}$.

The outline of this article is as follows: In section \ref{prelim}, we define a Hilbert manifold structure for the set of possible initial data, the phase space. In section \ref{Sconstraint}, we use an implicit function theorem argument to prove that the set of solutions to the constraint equations, is a Hilbert submanifold of the phase space; we call this the constraint submanifold. In section \ref{Shamiltonian} we define the energy, momentum and charge functionals on the phase space, and construct an appropriate Hamiltonian for this system, similar to that of Regge and Teitelboim \cite{RT}. This Hamiltonian ensures that Hamilton's equations give the correct evolution equations and on shell, it gives a value for the total energy of the system. In section \ref{Sfirstlaw}, we use a Lagrange multiplier argument to prove that stationarity is not only a sufficient condition for the first law of black hole thermodynamics to hold, it is necessary. Evidence for this is given in \cite{SW1}, however a rigorous proof requires the Hilbert manifold structure discussed in section \ref{Sconstraint}.

The phase space considered is tuples $(g,A,\pi,\varepsilon)$ with $H^2\times H^2\times H^1\times H^1$ local regularity, with appropriate decay conditions on $g,\pi$ for asymptotically flat spacetimes. The decay conditions on the fields $A$ and $\varepsilon$ (discussed in section \ref{prelim}) are more subtle; in addition to the requirement that the fields are asymptotically zero, we further require that $A$ approaches a collection of Maxwell (photon) fields at a faster rate. For simplicity, we work on a 3-manifold $\mathcal{M}$ with one asymptotic end and no interior boundary, however it is clear that all results will still hold in the case of many asymptotic ends. We will consider separately the case when $\mathcal{M}$ has an interior boundary.

The Hilbert manifold structure for the space of solutions is equivalent to the property of linearisation stability, which was studied by many authors throughout the 70s. Linearisation stability of Minkowski space was established by Choquet-Bruaht and Deser in 1973 \cite{C-BDstability} and in the same year, Fischer and Marsden proved linearisation stability for non-exceptional\footnote{Moncrief later proved that the exceptional data considered by Fischer and Marsden corresponds exactly to solutions exhibiting Killing fields \cite{moncrief1,moncrief2}} data on a compact manifold \cite{FM1}. This was subsequently extended to the Einstein-Maxwell \cite{armsEM} and Einstein-Yang-Mills \cite{armsEYM} cases by Arms. The general asymptotically flat case wasn't established until 2005 when Bartnik \cite{phasespace} provided a Hilbert manifold structure for the phase space for the Einstein equations. We follow the techniques of Bartnik to generalise this to the Einstein-Yang-Mills case. Like Bartnik, we consider a class of initial data too rough to guarantee that a solution to the constraints corresponds to a full solution, however linearisation stability can be obtained by noting that the analysis presented in section \ref{Sconstraint} remains valid if the phase space is required to have enough regularity for known existence and uniqueness theorems to be applied \cite{Klainerman05,RendallLRR}.

\section{Notation and Preliminary Definitions}\label{prelim}
Let $\mathcal{M}$ be a paracompact, connected, oriented and non-compact 3-manifold without boundary, and suppose there exists a compact $\mathcal{M}_0\subset\mathcal{M}$ and a diffeomorphism $\phi:\mathcal{M}\setminus \mathcal{M}_0\rightarrow\mathbb{R}^3\setminus\overline{B_1(0)}$, where $\overline{B_1(0)}$ is the closed unit ball. Let $\lie{g}$ be the Lie algebra of some compact Lie group $\hat{G}$, and recall that any such Lie algebra must be the direct sum of abelian and semi-simple Lie algebras. We can then define an adjoint invariant positive definite inner product, $\gamma$ on $\lie{g}$, given by the negative of the Killing form on the semi-simple factor; we may use the regular Euclidean inner product on the abelian factor.

Throughout this article, we use four different sets of indices on different objects as outlined below,
\begin{align*}
\begin{tabular}{|r||c    l|}
\hline
  $\mathcal{M}$, $\mathbb{R}^3$&Latin lower case, mid-alphabet &$i,j,...$\\\hline
  ${}^4\hspace{-0.6mm}\mathcal{M}$, $\mathbb{R}^{3,1}$ &Greek lower case, mid-alphabet  & $\mu,\nu...$  \\ \hline
  $\lie{g}$ &Latin lower case, early alphabet &$a,b...$    \\  \hline
  ${}^4\hspace{-0.6mm} P$, $(\mathbb{R}^{3,1}\oplus\lie{g})$ &Greek lower case, early alphabet &$\alpha,\beta...$     \\  \hline
\end{tabular}
\end{align*}

Fix a smooth background metric $\mathring{g}$ on $\mathcal{M}$, such that $\mathring{g}=\phi^*(\sigma)$ on $\mathcal{M}\setminus \mathcal{M}_0$, where $\sigma$ is the Euclidian metric on $\mathbb{R}^3$. In terms of this background metric, we define the weighted Lebesgue ($L^p_\delta$) and Sobolev ($W^{k,p}_\delta$) spaces as the completion of $C^\infty_c(\mathcal{M})$ with respect to the norms
\begin{align}
\left\|u\right\|_{p,\delta}&=
\left\{
\begin{array}{ll}
\left(\int_\mathcal{M}\left| u\right|^p r^{-\delta p-3}d\mu_o\right)^{1/p},& p<\infty\\
\text{ess sup}_\mathcal{M}(r^{-\delta}|u|), & p=\infty
\end{array}
\right.
\\
\left\|u\right\|_{k,p,\delta}&=\sum_{j=0}^k\|\mathring{\nabla}^j u\|_{p,\delta-j},
\end{align}
where $r\geq1$ is a smooth function on $\mathcal{M}$ and $r(x)=|\phi(x)|$ on $\mathcal{M}\setminus\mathcal{M}_0$ is the regular Euclidian distance from the origin in $\mathbb{R}^3$. Objects labeled with an `o' refer to objects associated with $\mathring{g}$, such as the connection $\onabla$ and volume form $d\mu_o$. Spaces of sections of bundles are defined on $\mathcal{M}$ in the regular way with respect to $\mathring{g}$ and $\gamma$ where appropriate and we will omit specifying the bundle where there is no risk of confusion. Intuitively, the spaces defined above contain functions of local regularity $L^p$ or $W^{k,p}$, which behave as $o(r^\delta)$ as $r\rightarrow\infty$, with derivatives decaying appropriately. Since we are working with a trivial bundle, we may choose the flat connection as a background and work in a gauge such that the background gauge covariant derivative is exactly $\mathring{\nabla}$.

The usual definition of Yang-Mills total charge \cite{globalcharges} is given by
\begin{equation}
4\pi Q^a:=\lim_{R\rightarrow\infty}\int_{S_R} *{}^4\hspace{-0.6mm}F^a\label{YMcharge},
\end{equation}
where $S_R:=\{x\in\mathcal{M}:r(x)=R\}$ is the sphere of Euclidean radius, $R$, for large $R$. Unfortunately this is gauge dependent in general; in fact, it may be that the integral is finite in some gauge and infinite in another. A sufficient condition to ensure the charge is well defined, is to ensure $[A_\mu,F^{\mu\nu}]\in L^1$ (see \cite{globalcharges}). In the Hamiltonian formulation, the choice of $A_0$ is still a gauge freedom, however we instead impose the condition that $[A_i,E^j]\in L^1$ for any possible choice of $E$ in the phase space.

Specifically, we will enforce that the dominant part of $A$ near infinity is valued in the centre of $\mathfrak{g}$. In the language of physics, this condition is that the gauge field behaves like a collection of photon fields near infinity.

In defining the following spaces, we will make use of the decomposition $\mathfrak{g}=\mathfrak{z}\oplus\mathfrak{k}$, where $\mathfrak{z}$ is the centre of $\mathfrak{g}$.
\begin{align*}
\mathcal{G}:&=W^{2,2}_{-1/2}(S_2)&
\mathcal{K}:&=W^{1,2}_{-3/2}(S^2\otimes\Lambda^3)\\
\mathcal{A}:&=W^{2,2}_{-1/2}(T^*\mathcal{M}\otimes \mathfrak{z})\oplus W^{2,2}_{-3/2}(T^*\mathcal{M}\otimes \mathfrak{k})&
\mathcal{E}:&=W^{1,2}_{-3/2}(T\mathcal{M}\otimes \mathfrak{g}^*\otimes\Lambda^3)\\
\mathcal{N}:&=L^2_{-1/2}(\Lambda^0\times T\mathcal{M}\times\mathfrak{g}\otimes\Lambda^0)&
\mathcal{N}^*:&=L^2_{-5/2}(\Lambda^3\times T^*\mathcal{M}\otimes\Lambda^3\times \mathfrak{g}^*\otimes\Lambda^3),
\end{align*}
where $\Lambda^k$ are $k$-forms on $\mathcal{M}$ and $S_2$ and $S^2$ are symmetric covariant and contravariant $2$-tensors on $\mathcal{M}$ respectively. The spaces $\mathcal{N}$ and $\mathcal{N}^*$ can be interpreted as spaces of vector fields and covector densities, respectively on ${}^4\hspace{-0.6mm} P$. The direct sum in the definition of $\mathcal{A}$ is understood as the internal sum in $W^{2,2}_{-1/2}(T^*\mathcal{M}\otimes \mathfrak{g})$. For an arbitrary $A=A_\mathfrak{z}+A_\mathfrak{k}\in\mathcal{A}$, we will write 
\begin{equation*}
\|A\|_\mathcal{A}:=\|A_\mathfrak{z}\|_{2,2,-1/2}+\|A_\mathfrak{k}\|_{2,2,-3/2}.
\end{equation*}
Note, if $\mathfrak{g}=\mathfrak{u}(1)=\mathfrak{z}$, then this includes the regular decay conditions ($E,B=O(r^{-2})$) for the Einstein-Maxwell equations.

Define the spaces
\begin{align*}
\mathcal{G}^+&=\left\{g:g-\mathring{g}\in\mathcal{G},g>0\right\}\\
\mathcal{G}^+_\lambda&=\left\{g\in\mathcal{G}^+:\lambda\mathring{g}<g<\lambda^{-1}\mathring{g}\right\},\qquad0<\lambda<1.
\end{align*}
From the weighted version of Morrey's inequality (\ref{Morrey}), we can deduce that both $g\in\mathcal{G}$ and $A\in\mathcal{A}$ are H\"older continuous with exponent $\frac{1}{2}$. In particular, the inequalities in the definitions of $\mathcal{G}^+$ and $\mathcal{G}^+_\lambda$ are understood in the pointwise sense.

The phase space we will consider is
\begin{equation}
\mathcal{F}:=\mathcal{G}^+\times\mathcal{A}\times\mathcal{K}\times\mathcal{E},
\end{equation}
which is independent of $\mathring{g}$ (and $\phi$) (see \cite{phasespace}). We now quote directly, the weighted Sobolev-type inequalities from \cite{AF}.
\begin{theorem}\label{ineqs} The following inequalities hold:
\begin{enumerate}[i.)]
\item If $1\leq p\leq q\leq\infty$, $\delta_2 < \delta_1$ and $u\in L^q_{\delta_2}$, then
\begin{equation}
\left\|u\right\|_{p,\delta_1}\leq c\left\|u\right\|_{q,\delta_2}\label{i}
\end{equation}
and thus $L^q_{\delta_2}\subset L^p_{\delta_1}$.
\item \textnormal{(H\"older)} If $u\in L^q_{\delta_1}$, $v\in L^r_{\delta_2}$ and $\delta=\delta_1+\delta_2$, $1\leq p,q,r\leq\infty$, then
\begin{equation}\label{holder}
\left\|uv\right\|_{p,\delta}\leq\left\|u\right\|_{q,\delta_1}\left\|v\right\|_{r.\delta_2},
\end{equation}
where $1/p=1/q+1/r$.
\item \textnormal{(Interpolation)} For any $\epsilon>0$, there is a $C(\epsilon)$ such that, for all $u\in W^{2,p}_\delta$
\begin{equation}\label{interpolation}
\left\|u\right\|_{1,p,\delta}\leq\epsilon\left\|u\right\|_{2,p,\delta}+C(\epsilon)\left\|u\right\|_{p,\delta},
\end{equation}
for $1\leq p\leq\infty$.
\item \textnormal{(Sobolev)} If $u\in W^{k,p}_\delta$, then
\begin{equation}\label{sobolev}
\left\|u\right\|_{np/(n-kp),\delta}\leq c\left\|u\right\|_{k,q,\delta}
\end{equation}
for $q$ satisfying $p\leq q\leq np/(n-kp)$.

If $kp>n$ then
\begin{equation}
\|u\|_{\infty,\delta}\leq c\|u\|_{k,p,\delta}\label{sobolev2}
\end{equation}
\item \textnormal{(Morrey's)} If $u\in W^{k,p}_\delta$ and $0<\alpha\leq k-n/p\leq 1$, then
\begin{equation}\label{Morrey}
\|u\|_{C^{0,\alpha}_\delta}\leq c\|u\|_{k,p,\delta},
\end{equation}
where the weighted H\"older norm is given by
\begin{align*}
\|u\|_{C^{0,\alpha}_\delta}:=&\sup_{x\in\mathcal{M}}\Big{(}r^{-\delta+\alpha}(x)\sup_{4|x-y|\leq r(x)}\frac{|u(x)-u(y)|}{|x-y|^\alpha}\Big{)}\\
&+\sup_{x\in\mathcal{M}}\left(r^\delta(x)|u(x)|\right)
\end{align*}

\item \textnormal{(Poincar\'e)} If $\delta<0$ and $1\leq p<\infty$, for any $u\in W^{1,p}_\delta$ we have
\begin{equation}\label{poincare}
\|u\|_{p,\delta}\leq c \|\mathring{\nabla}u\|_{p,\delta-1},
\end{equation}
\end{enumerate}
where $n=3$ is the dimension of $\mathcal{M}$.
\end{theorem}

\section{The Constraint Submanifold}\label{Sconstraint}
The constraint equations -- defining the constraint map, $\Phi$, for sufficiently smooth data -- are given by
\begin{alignat}{3}
\Phi_0(g,A,\pi,\varepsilon)&=(\frac{1}{2}(\pi^k_k)^2-\pi^{ij}\pi_{ij}-2(E^k_a E_k^a+B^k_a B^a_k))g^{-1/2}+R\sqrt{g} &&  \,\nonumber\\
&=(\frac{1}{2}(\pi^k_k)^2-\pi^{ij}\pi_{ij}-(\frac{1}{8}\varepsilon^k_a \varepsilon_k^a+2B^k_a B^a_k))g^{-1/2}+R\sqrt{g} &&=T_{00} \label{constraint1} \\
\Phi_i(g,A,\pi,\varepsilon)&=2\nabla^j\pi_{ij}-\varepsilon^j_a(\nabla_i A^a_j-\nabla_j A^a_i+C^a_{bc}A^b_iA^c_j) &&=T_{0i}\label{constraint2}\\
\Phi_a(g,A,\pi,\varepsilon)&=-\partial_j\varepsilon^j_a-C^c_{ab}A^b_j\varepsilon^j_c &&=j_a\label{gauss}
\end{alignat}
where $(T_{\mu 0},j_a)$ is some prescribed source. The quantity $B^i_a:=\frac{1}{2}\epsilon^{ijk}(\partial_j A_{ak}-\partial_k A_{aj} +C_{abc}A^b_j A^c_k)$ is the Yang-Mills magnetic field, as viewed by a Gaussian normal set of observers; $C^a_{bc}=C^a{}_{bc}$ are the structure constants of $\mathfrak{g}$ and $\epsilon^{ijk}$ (resp. $\epsilon_{ijk}$) is the completely antisymmetric tensor density with weight $1$ (resp. $-1$). Also note, since $\varepsilon$ is a vector density, we have $\nabla\cdot\varepsilon=\partial\cdot\varepsilon=\onabla\cdot\varepsilon$. It should be noted that the quantity $\varepsilon=-4E$ differs from the canonical momentum of Arms \cite{armsEYM} by a factor of 4, and is the negative of that used by Sudarsky and Wald \cite{SW1}; we will alternate between using $E$ and $\varepsilon$, wherever it is convenient. As usual, we have used natural units where $c=G=1$ and we have set the coupling constant to $\sqrt{4\pi}$ to agree with regular electromagnetic theory; specifically, our constraints are derived from the action
\begin{equation}
S=\int_\mathcal{M}(R-|F|^2).
\end{equation}
Note, this differs from the regular action by a factor of $16\pi$. If we were to use a different coupling constant, a factor would be present in the $|F|^2$ term, however this makes no difference to the arguments presented.

Since we are following arguments from \cite{phasespace}, it will be useful to define the pure gravitational constraint map
\begin{equation}
\Psi(g,\pi):=\begin{bmatrix}\Phi_0(g,A,\pi,\varepsilon)+2(E^k_a E_k^a+B^k_a B^a_k)g^{-1/2}\\\Phi_i(g,A,\pi,\varepsilon)+\varepsilon^j_a(\nabla_i A^a_j-\nabla_j A^a_i+C^a_{bc}A^b_iA^c_j)\end{bmatrix}=\begin{bmatrix}\Psi_0(g,\pi)\\\Psi_i(g,\pi)\end{bmatrix}.
\end{equation}

Throughout, we use $c$ or $C$ to denote some constant depending on $(\mathcal{M},\mathring{g})$ and other fixed parameters, which may vary from line to line. Where appropriate, we will make explicit the parameters on which these constants depend.

We first show $\Phi:\mathcal{F}\to\mathcal{N}^*$ is a smooth map of Hilbert manifolds.
\begin{proposition}\label{prop1} Suppose $(g,A,\pi,\varepsilon)\in\mathcal{G}^+_\lambda\times\mathcal{A}\times\mathcal{K}\times\mathcal{E}\subset\mathcal{F}$ for some fixed $\lambda>0$, then there exists a constant $c=c(\lambda)$ such that
\begin{align}
\left\|\Phi_0(g,A,\pi,\varepsilon)\right\|_{2,-5/2}&\leq c(1+\left\|g-\mathring{g}\right\|^2_{2,2,-1/2}+\left\|\pi\right\|^2_{1,2,-3/2}+\left\|\varepsilon\right\|^2_{1,2,-3/2}+\left\|A\right\|^4_\mathcal{A})\label{Phi1}\\
\left\|\Phi_i(g,A,\pi,\varepsilon)\right\|_{2,-5/2}&\leq c(\|\mathring{\nabla}\pi\|_{2,-5/2}+\|\mathring{\nabla}g\|_{1,2,-3/2}\left\|\pi\right\|_{1,2,-3/2}+\left\|\varepsilon\right\|_{1,2,-3/2}(1+\left\|A\right\|^2_{2,2,-1/2}))\label{Phi2}\\
\left\|\Phi_a(g,A,\pi,\varepsilon)\right\|_{2,-5/2}&\leq c\left\|\varepsilon\right\|_\mathcal{E}(1+\left\|A\right\|_\mathcal{A})\label{Phi3}
\end{align}
\end{proposition}
\begin{proof}
From \cite{phasespace} (Prop. 3.1) we have the bounds,
\begin{align}
\left\|\Psi_0(g,\pi)\right\|_{2,-5/2}&\leq c(1+\left\|g-\mathring{g}\right\|^2_{2,2,-1/2}+\left\|\pi\right\|^2_{1,2,-3/2})\\
\left\|\Psi_i(g,\pi)\right\|_{2,-5/2}&\leq c(\|\mathring{\nabla}\pi\|_{2,-5/2}+\|\mathring{\nabla}g\|_{1,2,-3/2}\left\|\pi\right\|_{1,2,-3/2})
\end{align}
thus we need only to bound the Yang-Mills terms.

Applying (\ref{holder}) and (\ref{sobolev}), we have the inequality
\begin{equation}\label{squaredineq}
\|u^2\|_{2,\delta}\leq \|u\|_{4,\delta/2}^2\leq c\left\|u\right\|_{1,2,\delta/2}^2.
\end{equation}
Using (\ref{squaredineq}), (\ref{i}), (\ref{holder}) and (\ref{sobolev}), we have
\begin{align*}
\|B^2\|_{2,-5/2}&\leq c(\| \mathring{\nabla}A\|_{1,2,-3/2}+\left\|[A,A]\right\|_{1,2,-3/2})^2\\
&\leq c(\|A\|_{2,2,-1/2}+\|A_\mathfrak{k}^2\|_{1,2,-3/2})^2\\
&\leq c(\|A\|_{2,2,-1/2}+\|A_\mathfrak{k}^2\|_{2,-3/2}+\|\mathring{\nabla}(A_\mathfrak{k})A_\mathfrak{k}\|_{2,-5/2})^2\\
&\leq c(1+\|A_\mathfrak{z}\|_{2,2,-1/2}+\|A_\mathfrak{k}\|^2_{1,2,-3/2}+\|\mathring{\nabla}A_\mathfrak{k}\|_{4,-5/4}\|A_\mathfrak{k}\|_{4,-5/4})^2\\
&\leq c(1+\|A_\mathfrak{z}\|_{2,2,-1/2}+\|A_\mathfrak{k}\|^2_{1,2,-3/2}+\|\mathring{\nabla}A_\mathfrak{k}\|_{1,2,-5/2}\|A_\mathfrak{k}\|_{1,2,-3/2})^2\\
&\leq c(1+\|A\|_\mathcal{A}^4)
\end{align*}
The $E^2$ term is clearly taken care of by (\ref{squaredineq}); combining these bounds with the definition of $\mathcal{G}^+_\lambda$, we have established (\ref{Phi1}).

Similarly, we have
\begin{align}
\|\Phi_i(g,A,\pi,\varepsilon)-\Psi_i(g,\pi)\|_{2,-5/2}&\leq c(\|\varepsilon\mathring{\nabla}A\|_{2,-5/2}+\|\varepsilon A_\mathfrak{k}^2\|_{2,-5/2})\\
&\leq c(\|\varepsilon\|_{4,-5/4}\|\mathring{\nabla}A\|_{4,-5/4}+\|\varepsilon\|_{4,-3/2}\|A_\mathfrak{k}\|^2_{8,-1/2})\\
&\leq c\|\varepsilon\|_{1,2,-3/2}(1+\|A\|_{2,2,-1/2}^2),
\end{align}
which establishes (\ref{Phi2}).

Finally we have
\begin{align}
\|\Phi_a(g,A,\pi,\varepsilon)\|_{2,-5/2}&\leq c(\|\varepsilon\|_{1,2,-3/2}+\|A_\mathfrak{k}\|_{4,-5/4}\|\varepsilon\|_{4,-5/4})\\
&\leq c(\|\varepsilon\|_{1,2,-3/2}+\|A_\mathfrak{k}\|_{1,2,-3/2}\|\varepsilon\|_{1,2,-3/2})
\end{align}
giving (\ref{Phi3}), and thus completing the proof.
\end{proof}

\begin{corollary}\label{smoothness}
$\Phi:\mathcal{F}\rightarrow\mathcal{N}^*$ is smooth.
\end{corollary}
\begin{proof}
It can be seen from Proposition \ref{prop1}, that $\Phi:\mathcal{F}\rightarrow\mathcal{N}^*$ is locally bounded. We note, $R$ can be expressed as a polynomial function in $g$, $g^{-1}$, $\mathring{\nabla}g$ and $\mathring{\nabla}^2g$, and therefore the constraint map can be considered a polynomial function in 12 variables,
\begin{equation}
\overline{\Phi}(g,g^{-1},\sqrt{g},1/\sqrt{g},\mathring{\nabla}g,\mathring{\nabla}^2g,\pi,\mathring{\nabla}\pi,\varepsilon,\mathring{\nabla}\varepsilon,A,\mathring{\nabla}A)=\Phi(g,A,\pi,\varepsilon).
\end{equation}
For positive definite matrices, the maps $g\mapsto\mathring{\nabla}g$, $g\mapsto\mathring{\nabla}^2g$, $A\mapsto\mathring{\nabla}A$, $g\mapsto\sqrt{g}$, etc. are smooth. Further, locally bounded polynomial functions are smooth (in the sense of Fr\'echet differentiability) (see \cite{hillephillips}, chapter 26), it follows that $\Phi$ is a smooth map of Hilbert manifolds.
\end{proof}

The linearisation of $\Phi$ at a point $G=(g,A,\pi,\varepsilon)\in\mathcal{F}$
\begin{equation}
D\Phi_G:\mathcal{G}\times\mathcal{A}\times\mathcal{K}\times\mathcal{E}\to \mathcal{N}^*\nonumber
\end{equation}
is given by
\begin{alignat}{3}
D\Phi_{0\,G}(h,b,p,f)&=&&\,(\pi^k_k\pi^{ij}-2\pi^{ik}\pi^j_k-2(E^i_a E^{aj}+B^i_aB^{aj}))h_{ij}g^{-1/2}\nonumber\\
\,&\,&&+(\frac{1}{2}\pi^{ij}\pi_{ij}-\frac{1}{4}(\pi^k_k)^2+(E^k_a E^a_k+B^k_aB^a_k))h^j_j g^{-1/2}\nonumber\\
\,&\,&&+(\frac{1}{2}h^k_k R-\Delta h^k_k+\nabla^i\nabla^j h_{ij}-R^{ij}h_{ij})\sqrt{g}-4\epsilon^{ijk}(\nabla_j b^a_k+C^a_{bc}A^b_jb^c_k)B_{ai}g^{-1/2}\nonumber\\
\,&\,&&+(p^k_k\pi^j_j-2\pi^{ij}p_{ij})g^{-1/2}-\frac{1}{4}f^i_a\varepsilon^a_ig^{-1/2}\label{Dphi0}\\
D\Phi_{i\,G}(h,b,p,f)&=&&\,2\nabla_j(\pi^{jk}h_{ik})-\pi^{jk}\nabla_i h_{jk}-\varepsilon^j_a(\nabla_i b^a_j-\nabla_j b^a_i + C^a_{bc}(A^b_ib^c_j+b^b_iA^c_j))\nonumber\\
\,&\,&&+2\nabla_jp^j_i-f^j_a(\nabla_iA^a_j-\nabla_jA_i^a+C^a_{bc}A^b_iA^c_j)\label{Dphii}\\
D\Phi_{a\,G}(h,b,p,f)&=&&-(C^c_{ab}\varepsilon^j_c b^b_j + \nabla_jf^j_a + C^c_{ab}f^j_cA^b_j).\label{Dphia}
\end{alignat}
See \cite{ehlers} (and references therein) for computations.

The $L^2$ adjoint is simply computed by integration by parts and throwing out the boundary terms:
\begin{alignat}{3}
D\Phi_{g\,G}^*(N,X,V)&=&&\,N\Big{(}\pi^k_k\pi^{ij}-2\pi^{ik}\pi^j_k-2(E^i_a E^{aj}+B^i_aB^{ai})\nonumber\\
\,&\,&&+\left\{\frac{1}{2}\pi^{kl}\pi_{kl}-\frac{1}{4}(\pi^k_k)^2+(E^k_a E^a_k+B^k_aB_k^a)\right\} g^{ij}\Big{)}g^{-1/2}\nonumber\\
\,&\,&&+\left\{ N(\frac{1}{2}Rg^{ij}-R^{ij})+\nabla^i\nabla^j N-g^{ij}\nabla^k\nabla_k N\right\}\sqrt{g}+\mathcal{L}_X\pi^{ij}\label{adjointg}\\
D\Phi_{A\,G}^*(N,X,V)&=&&-4\epsilon^{ijk}\left\{\nabla_j(NB_{ak} g^{-1/2})+C_{abc}NB^c_kA^b_j g^{-1/2}\right\}+\mathcal{L}_X \varepsilon^i_a\nonumber\\
\,&\,&&-X^i(\nabla_j\varepsilon^j_a+C^c_{ab}A^b_j\varepsilon^j_c)+X^jC^c_{ab}A^b_j\varepsilon^i_c-C^b_{ca}\varepsilon^i_bV^c\label{adjointA}\\
D\Phi_{\pi\,G}^*(N,X,V)&=&&\,N(g_{ij}\pi^k_k-2\pi_{ij})g^{-1/2}-\mathcal{L}_X g_{ij}\label{adjointpi}\\
D\Phi_{\varepsilon\,G}^*(N,X,V)&=&&-\frac{1}{4}N\varepsilon^a_ig^{-1/2}+X^j(\nabla_i A^a_j-\nabla_j A^a_i +C^a_{bc}A^b_iA^c_j)+\partial_iV^a+C^a_{bc}A^b_i V^c\nonumber\\
&=&&\,N E^a_ig^{-1/2}-\epsilon_{ijk}X^jB^{ak}+\partial_iV^a+C^a_{bc}A^b_iV^c.\label{adjointeta}
\end{alignat}
The tuple $(N,X,V)\in\mathcal{N}$ corresponds to a scalar function, vector field and $\mathfrak{g}$-valued function on $\mathcal{M}$ respectively. We will omit reference to the base point $G$ when there is no risk of confusion.

It should be noted that the map given by
\begin{equation}
T(N,X):=\begin{bmatrix}D\Phi_g^*(N,X,V)\\D\Phi_\pi^*(N,X,V)\end{bmatrix}
\end{equation}
is of the exact same form considered in \cite{phasespace}, so we quote the following two theorems.
\begin{theorem}\label{weakstrong}
If $(f_1,f_3)\in L^2_{-3/2}\times W^{1,2}_{-3/2}$ and $(N,X)\in L^2_{-1/2}$ is a weak solution of $T(N,X)=(f_1,f_3)$, then $(N,X)\in W^{2,2}_{-1/2}$ is a strong solution.
\end{theorem}
\begin{theorem}\label{kernel}
The operator $T$ has trivial kernel in $L^2_{-1/2}$.
\end{theorem}

From this we will prove $D\Phi^*$ also has trivial kernel, but first we will need to prove an estimate.
\begin{lemma}\label{Vboundlemma}
If $\xi=(N,X,V)\in W^{2,2}_{-1/2}$ satisfies $D\Phi^*_\varepsilon(\xi)=f_4\in W^{1,2}_{-3/2}$, then
\begin{equation}
\|V\|_{2,2,-1/2}\leq c\Big{(}\|(N,X)\|_{2,2,0}+\|f_4\|_{1,2,-3/2}+\|V\|_{2,0}\Big{)}\label{Vbound},
\end{equation}
where $C$ depends on $(g,A,\pi,\varepsilon)$.
\end{lemma}
\begin{proof}
From (\ref{adjointeta}) we have
\begin{equation}
\partial_iV^a=D\Phi_\varepsilon^*(\xi)-NE^a_ig^{-1/2}+\epsilon_{ijk}X^jB^{ak}-C^a_{bc}A^b_iV^c\label{diffV}
\end{equation}
For brevity, let $\zeta=(N,X)\in W^{2,2}_{-1/2}$ and $\beta=(E,B)\in W^{1,2}_{-3/2}$. By differentiating (\ref{diffV}) and applying the inequalities from Theorem \ref{ineqs}, we have
\begin{align*}
\|\mathring{\nabla}^2V\|_{2,-5/2} &\leq c\Big{(}\|f_4\|_{1,2-3/2}+\|\onabla(\zeta)\beta\|_{2,-5/2}+\|\onabla(\beta)\zeta\|_{2,-5/2}\\
&\hspace{4mm}+\|\onabla(A_\lie{k})V\|_{2,-5/2}+\|\onabla(V)A_\lie{k}\|_{2,-5/2}\Big{)}\\
&\leq c\Big{(}\|f_4\|_{1,2-3/2}+\|\onabla(\zeta)\|_{4,-1}\|\beta\|_{4,-3/2}+\|\onabla(\beta)\|_{2,-5/2}\|\zeta\|_{\infty,0}\\
&\hspace{4mm}+\|\onabla A_\lie{k}\|_{4,-5/2}\|V\|_{4,0}+\|\onabla V\|_{2,-1}\|A_\lie{k}\|_{\infty,-3/2}\Big{)}\\
&\leq c\Big{(}\|f_4\|_{1,2-3/2}+\|\onabla(\zeta)\|_{1,2,-1}\|\beta\|_{1,2,-3/2}+\|\beta\|_{1,2,-3/2}\|\zeta\|_{2,2,0}\\
&\hspace{4mm}+\|\onabla A_\lie{k}\|_{1,2,-5/2}\|V\|_{1,2,0}+\|\onabla V\|_{2,-1}\|A_\lie{k}\|_{2,2,-3/2}\Big{)}\\
&\leq c\Big{(}\|f_4\|_{1,2-3/2}+\|\beta\|_{1,2,-3/2}\|\zeta\|_{2,2,0}+\|V\|_{1,2,0}\|A_\lie{k}\|_{2,2,-3/2}\Big{)}\\
&\leq c(\|f_4\|_{1,2-3/2}+\|\zeta\|_{2,2,0}+\|V\|_{1,2,0}),
\end{align*}
where the constant, $c$, depends on $(\beta,A_\lie{k})$ in the last line. Applying the weighted Poincar\'e inequality (\ref{poincare}), we have
\begin{equation}
\|V\|_{2,2,-1/2}\leq c\Big{(}\|(N,X)\|_{2,2,0}+\|f_4\|_{1,2,-3/2}+\|V\|_{1,2,0}\Big{)}.
\end{equation}
Applying the interpolation inequality (\ref{interpolation}) to the last term on the right hand side and choosing $\epsilon$ small enough, gives us (\ref{Vbound}).
\end{proof}

\begin{theorem}\label{weakstrongYM}
If $\xi\in\L^2_{-1/2}$ is a weak solution of $D\Phi^*(\xi)=(f_1,f_2,f_3,f_4)$, where $(f_1,f_3,f_4)\in L^2_{-3/2}\times W^{1,2}_{-3/2}\times W^{1,2}_{-3/2}$, then $\xi\in W^{2,2}_{-1/2}$.
\end{theorem}
\begin{proof}
From Theorem \ref{weakstrong}, we have $(N,X)\in W^{2,2}_{-1/2}$. From (\ref{adjointeta}), it can be seen that on each $\Omega\subset\subset\mathcal{M}$, $V$ (weakly) satisfies an equation of the form
\begin{equation}
\partial_i V^a=\alpha_{ic}^aV^c+\beta_i^a
\end{equation}
with coefficients $\alpha\in W^{2,2}$, $\beta\in W^{1,2}$ on $\Omega$. The argument here is well known; $V$ can be approximated by $V_\epsilon$, where $\epsilon$ is a mollification parameter and on $\Omega$ we have
\begin{equation}
\|\onabla V_\epsilon\|_2\leq c(\|\alpha\|_\infty\|V_\epsilon\|_2+\|\beta\|_2)\leq c(\|\alpha\|_{2,2}\|V\|_2+\|\beta\|_2).
\end{equation}
Since $V_\epsilon$ is uniformly bounded in $W^{1,2}$, it follows $V_\epsilon\rightharpoonup V\in W^{1,2}$.

By differentiating (\ref{adjointeta}), it can be seen that on any $\Omega\subset\subset\mathcal{M}$, $V$ weakly satisfies an equation of the form
\begin{equation}
\partial^2_{ij}V^a=\alpha_{ic}^a\partial_j V^c+\beta_{ijc}^aV^c+\theta_{ij}^{a},
\end{equation}
where here we have $\alpha\in W^{2,2}$, $\beta\in W^{1,2}$, $\theta\in L^2$. Since we now have $V\in W^{1,2}$, the same argument will give us $V_\epsilon\rightharpoonup V\in W^{2,2}(\Omega)$. All that remains to show, is that $V$ and its weak derivatives satisfy the correct asymptotic conditions.

For any smooth cutoff function $\chi_R$ with $\chi_R\equiv 1$ on $B_R(0)$ and zero outside $B_{2R}(0)$, $\xi=(N,X,\chi_R V)$ satisfies the conditions for lemma \ref{Vboundlemma} and thus we have
\begin{equation}
\|\chi_R V\|_{2,2,-1/2}\leq C\Big{(}\|(N,X)\|_{2,2,-1/2}+\|f_4\|_{1,2,-3/2}+\|V\|_{0,2,0}\Big{)}.
\end{equation}
Once more we have a uniform bound, thus it follows $\chi_R V\rightharpoonup V$ in $W^{2,2}_{-1/2}.$
\end{proof}
Next we demonstrate that $D\Phi^*$ has trivial kernel. This amounts to saying that there are no symmetries of the data, asymptotic to zero at infinity -- this will be discussed in more detail in section \ref{Shamiltonian}.
\begin{proposition}\label{trivialkernel}
If $\xi\in L^2_{-1/2}$ satisfies $D\Phi^*(\xi)\equiv 0$ on $\mathcal{M}$, then $\xi\equiv 0$.
\end{proposition}
\begin{proof}
From theorems \ref{kernel} and \ref{weakstrongYM} respectively, $(N,X)\equiv 0$ and $V\in W^{2,2}_{-1/2}$. From (\ref{adjointeta}), we have $\partial_i V^a=C^a_{bc}V^b A^c_i$ and we can repeat the arguments with $V_\epsilon$ and easily obtain the uniform bound,
\begin{align*}
\|V_\epsilon\|_{3,2,-1/2}&\leq c\|\partial^3 V_\epsilon\|_{2,-7/2}\\
&\leq c (\|\onabla^2 A_\lie{k}\|_{2,-7/2}\|V\|_{\infty,0}+\|\onabla A_\lie{k}\|_{4,-5/2}\|\onabla V\|_{4,-1}+\|A_\lie{k}\|_{\infty,-3/2}\|\onabla^2 V\|_{2,-2})\\
&\leq c\|A_\lie{k}\|_{2,2,-3/2}\|V\|_{2,2,0}.
\end{align*}
As above, we now have $V\in W^{3,2}_{-1/2}$; from the weighted version of Morrey's inequality (\ref{Morrey}), we have $V\in C^{1,1/2}$. That is $V$ strongly satisfies the equation
\begin{equation}
\frac{1}{2}\partial_i(V^aV_a)=\partial_i(V^a)V_a=C^a_{bc}V^bA_i^cV_a=0;
\end{equation}
since $V$ is asymptotic to zero and $\mathcal{M}$ is connected, $V\equiv 0$.
\end{proof}

We are now able to apply the implicit function theorem to prove the level sets of $\Phi$ are smooth submanifolds of $\mathcal{F}$. 
\begin{theorem}
For any $(s,S_i,\sigma_a)\in\mathcal{N}^*$, the set
\begin{equation}
\mathcal{C}(s,S,\sigma):=\{(g,A,\pi,\varepsilon)\in\mathcal{F}:\Phi(g,A,\pi,\varepsilon)=(s,S,\sigma)\}
\end{equation}
is a Hilbert submanifold of $\mathcal{F}$.
\end{theorem}
\begin{proof}
We simply must establish that $D\Phi$ is surjective and splits its domain into the direct sum of the $Ker(D\Phi)$ and a complementary subspace, then the result follows from the implicit function theorem. Since $D\Phi$ is bounded, the kernel is closed and hence splits. The codomain splits as $\mathcal{N}^*=\overline{\Ran(D\Phi)}\oplus \coker(D\Phi)$, but from proposition \ref{trivialkernel}, $\coker(D\Phi)$ is trivial. To establish surjectivity, we simply must show that $D\Phi$ has closed range.

Consider variations of the form
\begin{align}
h_{ij}=-\frac{1}{2}g_{ij}y&&b_i^a=0\\
p^{ij}=\frac{1}{2}(\nabla^iY^j+\nabla^iY^j-\nabla_kY^kg^{ij})\sqrt{g}&&f_i^a=-\partial_i\psi^a\sqrt{g}.
\end{align}
With $(h,b,p,f)$ of this form, define $F(\mathcal{Y})=F(y,Y,\psi)=D\Phi_{(g,A,\pi,\varepsilon)}(h,0,p,f)$. Explicitly,
\begin{align*}
F&(y,Y,\psi)=\\
&\begin{bmatrix}
\Delta y\sqrt{g}-\frac{1}{4}\Phi_0(g,A,\pi,\varepsilon)y+\frac{1}{2}\pi^j_j\nabla_j Y^j-2\pi^{ij}\nabla_i Y_j-\frac{1}{4}(E^2+B^2)y+\varepsilon^i_a\partial_i\psi^a\\
\Delta Y_i \sqrt{g}+R_{ij}Y^j\sqrt{g}+\onabla^j(\psi_a)(\onabla_i A^a_j-\onabla_j A_i^a+C^a_{bc}A^b_iA^c_j)\sqrt{g}-\nabla_j(\pi^j_i)y-\pi_i^j\onabla_jy+\frac{1}{2}\pi^j_j\onabla_iy\\
\mathring{\Delta}\psi_a\sqrt{g}+C^c_{ab}\onabla^j(\psi_c)A^b_j\sqrt{g}
\end{bmatrix}
\end{align*}
Clearly we have $F:W^{2,2}_{-1/2}\rightarrow\mathcal{N}^*$. This new operator is clearly bounded and the adjoint map has similar structure. We have the following scale-broken estimate from \cite{AF}
\begin{equation}
\|u\|_{2,2,-1/2}\leq c(\|\Delta u\|_{2,-5/2}+\|u\|_{2,0}),
\end{equation}
from which we can establish an elliptic estimate for $F$.
\begin{align}
\|\Delta\mathcal{Y}\|_{2,-5/2}&\leq c\Big{(}\|F(\mathcal{Y})\|_{2,-5/2}+\|\Phi_0 y\|_{2,-5/2}+\|\pi\onabla \mathcal{Y}\|_{2,-5/2}+\|\pi\tilde{\Gamma}Y\|_{2,-5/2}+\|Ric(Y)\|_{2,-5/2}\nonumber\\
&+\|\onabla(\psi)\onabla{A}\|_{2,-5/2}+\|\onabla(\psi)A^2_\lie{k}\|_{2,-5/2}+\|\nabla(\pi)y\|_{2,-5/2}+\|\onabla(\psi)A_\lie{k}\|_{2,-5/2}\nonumber\\
&+\|E^2 y\|_{2,-5/2}+\|B^2 y\|_{2,-5/2}+\|E\onabla\psi\|_{2,-5/2}\Big{)}
\end{align}
Where $\tilde{\Gamma}^k_{ij}:=\mathring{\Gamma}^k_{ij}-\Gamma^k_{ij}=\frac{1}{2}g^{kl}(\onabla_ig_{jl}+\onabla_jg_{il}-\onabla_lg_{ij})$ is the connection difference tensor and is clearly $W^{1,2}_{-3/2}$. It's easy to check that $Ric$ is of the form $Ric\sim(\mathring{Ric}+\onabla\tilde{\Gamma}+\tilde{\Gamma}^2)$, so it follows $Ric\in L^2_{-5/2}$.

For the sake of presentation, we define the quantities $U_1:=(\Phi_0,Ric,\nabla\pi,\pi\tilde{\Gamma},E^2,B^2)\in L^2_{-5/2}$ and ${U_2:=(\pi,\onabla(A),A_\lie{k},A^2_\lie{k},E)\in W^{1,2}_{-3/2}}$. With this notation we have
\begin{equation}
\|\mathcal{Y}\|_{2,2,-1/2}\leq c(\|F(\mathcal{Y})\|_{2,-5/2}+\|U_1 \mathcal{Y}\|_{2,-5/2}+\|U_2\onabla\mathcal{Y}\|_{2,-5/2}+\|\mathcal{Y}\|_{2,0}).\label{firstFbound}
\end{equation}
Now the separate terms can be easily bound
\begin{align*}
\|U_1\mathcal{Y}\|_{2,-5/2}&\leq \|U_1\|_{2,-5/2}\|\mathcal{Y}\|_{\infty,0}\\
&\leq C\|\mathcal{Y}\|_{1,4,0}\\
&\leq C(\|\mathcal{Y}\|_{4,0}+\|\onabla\mathcal{Y}\|_{4,-1})\\
&\leq C(\|\mathcal{Y}^{1/4}\mathcal{Y}^{3/4}\|_{4,0}+\|\onabla\mathcal{Y}^{1/4}\onabla\mathcal{Y}^{3/4}\|_{4,-1})\\
&\leq C(\|\mathcal{Y}\|_{2,0}^{1/4}\|\mathcal{Y}\|^{3/4}_{6,0}+\|\onabla\mathcal{Y}\|_{2,-1}^{1/4}\|\onabla\mathcal{Y}\|_{6,-1}^{3/4})\\
&\leq C\|\mathcal{Y}\|_{1,2,0}^{1/4}\|\mathcal{Y}\|^{3/4}_{2,2,0}\\
&\leq \epsilon\|\mathcal{Y}\|_{2,2,0}+\frac{C}{\epsilon^3}\|\mathcal{Y}\|_{1,2,0},
\end{align*}
where the last line comes from Young's inequality.

Similarly, we have
\begin{align*}
\|U_2\onabla\mathcal{Y}\|_{2,-5/2}&\leq c\|U_2\|_{6,-3/2}\|\onabla\mathcal{Y}\|_{3,-1}\\
&\leq c\|U_2\|_{1,2,-3/2}\|\onabla\mathcal{Y}\|_{3,-1}\\
&\leq C\|\onabla(\mathcal{Y})^{1/3}\onabla(\mathcal{Y})^{2/3}\|_{3,-1}\\
&\leq C\|\onabla\mathcal{Y}\|^{1/3}_{2,-1}\|\onabla\mathcal{Y}\|^{2/3}_{4,-1}\\
&\leq C\|\onabla\mathcal{Y}\|^{1/3}_{2,-1}\|\onabla\mathcal{Y}\|^{2/3}_{1,2,-1}\\
&\leq \epsilon\|\mathcal{Y}\|_{2,2,0}+\frac{C}{\epsilon^2}\|\mathcal{Y}\|_{1,2,0},
\end{align*}
where we have switched from $c$ to $C$ to indicate the constant's dependence on $\|(g,A,\pi,\varepsilon)\|_\mathcal{F}$.

Combining these estimates with (\ref{firstFbound}) and applying (\ref{interpolation}) to $\|\mathcal{Y}\|_{1,2,0}$, we have
\begin{equation}
\|\mathcal{Y}\|_{2,2,-1/2}\leq c (\|F(\mathcal{Y})\|_{2,-5/2}+\|\mathcal{Y}\|_{2,0}).\label{Fbound}
\end{equation}
By construction, the adjoint operator $F^*$ has the same structure and thus also satisfies an estimate of the form (\ref{Fbound}). In particular, this implies $\ker(F^*)\subset W^{2,2}_{-1/2}$. Take a sequence $\mathcal{X}_n\in \ker(F^*)$ with $\|\mathcal{X}_n\|_{2,2,-1/2}\leq 1$ and we have $\|\mathcal{X}_n-\mathcal{X}_m\|_{2,2,-1/2}\leq C\|\mathcal{X}_n-\mathcal{X}_m\|_{2,0}$ and by passing to a subsequence and applying a weighted version of the Rellich compactness theorem (see \cite{ellipticsys}, Lemma 2.1), the closed unit ball in $\ker(F^*)$ is compact and thus $\ker(F^*)$ is finite dimensional. The same reasoning tells us that $\ker(F)$ is also finite dimensional, thus there is a closed subspace $Z$ such that $W^{2,2}_{-1/2}=Z\oplus \ker(F)$. To show $F$ has closed range, it will suffice to prove
\begin{equation}
\|\mathcal{Y}\|_{2,2,-1/2}\leq C\|F(\mathcal{Y})\|_{2,-5/2}, \qquad \text{for all }\mathcal{Y}\in Z.\label{Z}
\end{equation}
If (\ref{Z}) were not true, we could take a sequence $\mathcal{Y}_n$ in $Z$ with $\|\mathcal{Y}_n\|_{2,2,-5/2}=1$ such that $\|F(\mathcal{Y}_n)\|_{2,-5/2}\rightarrow0$. Then we can again pass to a subsequence converging in $L^2_0$ and (\ref{Fbound}) implies $\mathcal{Y}_n$ converges to some $\mathcal{Y}\neq 0$ in $Z\cap\ker(F)$, which would be a contradiction. Now we have $\ker(F)\oplus\coker(F)=\mathcal{N}^*$.

Clearly $\Ran(F)\subset\Ran(D\Phi)$, so at most $\Ran(D\Phi)$ differs from $\mathcal{N}^*$ by a finite dimensional closed subspace, and since $\overline{\Ran(D\Phi)}=\mathcal{N}^*$, $D\Phi$ is surjective.
\end{proof}

\section{The Hamiltonian}\label{Shamiltonian}
It is well known, that in order to generate the correct equations of motion, the first variation of the Hamiltonian density must be of the form
\begin{equation}
\delta H=X\cdot\delta q+Y\cdot\delta p,
\end{equation}
where $q$ and $p$ are the canonical position and momentum respectively. Hamilton's equations can then be read off as
\begin{align}
\frac{\partial q}{\partial t}=Y,
\hspace{4mm}&\hspace{4mm}\frac{\partial p}{\partial t}=-X,
\end{align}
where $t$ is the time parameter.

In the framework of general relativity we need to make precise what we mean by `time'. We interpret $t$ as the flow parameter of the (yet to be specified) lapse-shift vector field on the spacetime. See \cite{ehlers} for a detailed discussion on this.

In the Einstein-Yang-Mills case, we have a vector field $\xi$ on the bundle ${}^4\hspace{-0.6mm}P$ generating the evolution, in place of the usual lapse-shift vector field. This corresponds to both a choice of coordinates on $^4\mathcal{M}$ and a choice of gauge. One may interpret the flow of $\xi$ as simultaneously evolving the data through time while continuously changing the gauge. See \cite{armsEYM} for details.

In the Einstein-Yang-Mills case, the canonical position and momentum are $(g,A)$ and $(\pi,\varepsilon)$ respectively. In order to generate the correct equations of motion, we should expect to write the first variation of the Hamiltonian density as
\begin{align}
\delta H(g,A,\pi,\varepsilon)=&X_1(g,A,\pi,\varepsilon;\xi)^{ij} \delta g_{ij}+X_2(g,A,\pi,\varepsilon;\xi)_a^i \delta A_i^a \\
&+Y_1(g,A,\pi,\varepsilon;\xi)_{ij} \delta \pi^{ij}+Y_2(g,A,\pi,\varepsilon;\xi)^a_i \delta \varepsilon_a^i,
\end{align}
or equivalently
\begin{equation}
D H_{(g,A,\pi,\varepsilon)}\cdot(h,b,p,f)=(X_1,X_2,Y_1,Y_2)|_{(g,A,\pi,\varepsilon)}\cdot(h,b,p,f).\label{hamform}
\end{equation}
Notice that the Hamiltonian must have some dependence on the direction $\xi$ in ${}^4\hspace{-0.6mm}P$, in which we are to evolve the data.

For the remainder of this paper, it will be convenient to write a point in the phase space as $G=(g,A,\pi,\varepsilon)\in\mathcal{F}$ and a tangent vector $Z=(h,b,p,f)\in T_G\mathcal{F}$. With this notation, the usual ADM Hamiltonian (with Yang-Mills fields) is given by
\begin{equation}
\mathcal{H}^{ADM}(G;\xi)=-\int_\mathcal{M}\xi^\alpha\Phi_\alpha(G),
\end{equation}
this is the pure constraint form of the Hamiltonian.

It will be shown, if the boundary terms which were cast out when defining $D\Phi^*$ do indeed vanish, then we have
\begin{equation}
D\mathcal{H}^{ADM}_{G}(\xi)\cdot Z=-\int_\mathcal{M}D\Phi_{G}^*(\xi)\cdot Z,\label{DHADM}
\end{equation}
and therefore the Hamiltonian density is of the form (\ref{hamform}). We will see however, that if $\xi$ is not asymptotically zero, then this will not be the case; we will discuss this point later.

Hamilton's equations now become
\begin{equation}\label{evoeq}
\frac{\partial}{\partial t}
\begin{bmatrix}
g\\A\\ \pi\\ \varepsilon
\end{bmatrix}
=-J\circ D\Phi_{G}^*(\xi),
\end{equation}
where
\begin{equation}
J=\begin{bmatrix}0&0&1&0\\0&0&0&1\\-1&0&0&0\\0&-1&0&0\end{bmatrix},
\end{equation}
is the natural symplectic structure on $\mathcal{F}$. Equation (\ref{evoeq}) motivates Moncrief's result \cite{moncrief2} equating elements of $\ker (D\Phi_G^*)$ with spacetime Killing vector fields. We thus refer to elements of $\ker (D\Phi_G^*)$ as \textit{generalised Killing vectors} and if $\xi$ corresponds to a stationary Killing field, then we call $G=(g,A,\pi,\varepsilon)$ \textit{generalised stationary data}.
\begin{proposition}
The map $\mathcal{H}^{ADM}:\mathcal{F}\times\mathcal{N}\rightarrow\mathbb{R}$ is a smooth map.
\end{proposition}
\begin{proof}
The smoothness in $G$ follows from the smoothness of $\Phi$. We have $|\mathcal{H}^{ADM}(G;\xi)|=\|\xi\Phi\|_{1,-3}\leq\|\xi\|_{2,-1/2}\|\Phi\|_{2,-5/2}$, that is, $\mathcal{H}^{ADM}$ is bounded and linear in $\xi$.
\end{proof}
We next establish the validity of equation (\ref{DHADM}).

\begin{theorem}\label{admhamil}
For all $\xi\in W^{2,2}_{-1/2}$,
\begin{equation}
D\mathcal{H}^{ADM}_{G}(\xi)\cdot Z=-\int_\mathcal{M}D\Phi_{G}^*(\xi)\cdot Z\nonumber,
\end{equation}
for all $Z\in T_{G}\mathcal{F}$.
\end{theorem}
\begin{proof}
This is equivalent to the statement that the formal adjoint of $D\Phi_G$, given in section \ref{Sconstraint}, is indeed the adjoint. We simply must demonstrate the boundary terms at infinity arising from integration by parts do indeed vanish. These boundary terms are given by
\begin{align}
Z\cdot& D\Phi_G^*(\xi)-\xi^\alpha D\Phi_{G\alpha}(Z)=\nonumber\\
&\nabla^i\Big{(}(N(\onabla_i\text{tr}_gh-\nabla^jh_{ij})+\onabla^j(N)h_{ij}-\text{tr}_g h\onabla_i(N))\sqrt{g}-2X^jp_{ij}+V^af_{ai}\Big{)}\label{boundary}\\
&-\nabla^i\Big{(}2\pi^k_i h_{jk}X^j - \pi^{jk}h_{jk}X_i +\epsilon_{ijk}b^{ak}B^j_a N\sqrt{g}+\varepsilon_{ia}b^a_jX^j-X_i\varepsilon^j_ab^a_j\Big{)}\nonumber.
\end{align}
The boundary terms have been expressed as two separate divergences corresponding to their decay rates at infinity - this distinction will be important later. Note, these divergences do indeed make sense as boundary integrals at inifinity in the usual (trace) sense, (see \cite{phasespace} - Lemma 4.3 and Lemma 4.4). Formally, we consider an exhaustion of $\mathcal{M}$ by compact sets $\mathcal{M}_k$ with smooth boundary, and take the limit of the boundary integrals as $k\rightarrow\infty$. For convenience, we choose the exhaustion to be euclidean balls near infinity and consider the limit of integrals on spheres.

Lemma 4.4 of \cite{phasespace} gives us the estimate
\begin{equation}
\oint_{S_R}|u|\leq c \sqrt{R}\|u\|_{1,2,-3/2:A_R}\label{tracebound},
\end{equation}
where $S_R$ is the sphere of radius $R$ centred at zero, and $A_R$ is the region bound between $S_{R}$ and $S_{2R}$. For simplicity, let us denote by $\nabla^i\mathcal{B}^1_{i}$ and $\nabla^i\mathcal{B}^2_{i}$, the first and second divergences in (\ref{boundary}) respectively. $\mathcal{B}^2$ is a collection of terms of the form $\alpha\beta\xi$; where $\alpha\in W^{2,2}_{-1/2}$, $\beta\in W^{1,2}_{-3/2}$ and $\xi\in W^{2,2}_{-1/2}$;. Note, $g$ and $g^{-1}$ are bound, so we needn't consider the raising or lowering of indices in our estimates.

Applying (\ref{tracebound}), we have
\begin{align}
\oint_{S_R}|\mathcal{B}^2|&\leq c \|\alpha\|_{\infty:S_R}\|\xi\|_{\infty:S_R}\|\beta\|_{1:S_R}\nonumber\\
&\leq o(R^{-1/2})\|\xi\|_{\infty:S_R}\sqrt{R}\|\beta\|_{1,2,-3/2}\nonumber\\
&\leq o(1)\|\xi\|_{\infty:S_R}\|\beta\|_{1,2,-3/2}\label{B2bound},
\end{align}
where we have made use of the fact $\alpha\in W^{2,2}_{-1/2}(\mathcal{M})\subset C^0(\mathcal{M})$. In the limit as $R$ tends to infinity, this integral vanishes and therefore $\mathcal{B}^2$ contributes no boundary terms. Note that this still holds if $\xi$ is only $C^0$ and bound. $\mathcal{B}^1$ can be expressed as a collection of terms of the form $\xi\tau$, where $\xi\in W^{2,2}_{-1/2}$ and $\tau\in W^{1,2}_{-3/2}$.
\begin{align}
\oint_{S_R}|\mathcal{B}^1|&\leq c\|\xi\|_{\infty:S_R}\|\tau\|_{1:S_R}\\
&\leq o(1)\|\tau\|_{1,2,-3/2}
\end{align}
For the same reasons as above, the remaining boundary terms also vanish and therefore we have  $\int_\mathcal{M}\xi\cdot D\Phi_{G}(Z)=\int_\mathcal{M}Z\cdot D\Phi_{G}^*(\xi)$.\\
\end{proof}
The necessity that $\xi\rightarrow 0$ at infinity is twofold; not only does it ensure we have control on $\mathcal{B}^1$, it is also required for $\mathcal{H}^{ADM}$ to be well defined on $\mathcal{F}$, as $\Phi(G)$ is not integrable for generic initial data. When $\xi$ is taken to be asymptotic to some non-zero constant vector, it will be shown that the non-vanishing boundary terms, $\mathcal{B}^1$, correspond to the first variation of energy-momentum and charge. This leads us to modify our Hamiltonian \`a la Regge and Teitelboim \cite{RT}. Before discussing this, we should make precise what we mean by ``asymptotic to a constant vector".

Fix some $\xi_\infty\in\mathbb{R}^{3,1}\oplus\mathfrak{g}$, which on some exterior region $E_R$, may be identified with a section, ${\tilde{\xi}_\infty\in C^\infty (\Lambda^0\times T\mathcal{M}\times\mathfrak{g}\otimes\Lambda^0)}$, such that $\onabla\xi=0$. We now represent $\xi_\infty$ as $\hat{\xi}_\infty\in (\Lambda^0\times T\mathcal{M}\times\mathfrak{g}\otimes\Lambda^0)$, with $\hat{\xi}_\infty=\tilde{\xi}_\infty$ on $E_{2R}$ and $\hat{\xi}\equiv 0$ on $E_R$. Obviously $\hat{\xi}_\infty$ is not unique, however the difference between any two choices of $\hat{\xi}_\infty$ is in $C^\infty_0 (\Lambda^0\times T\mathcal{M}\times(\mathfrak{g}\otimes\Lambda^0))\subset\mathcal{N}$. This means that the space 
\begin{equation}
\xi_\infty+\mathcal{N}:=\{\xi:\xi-\hat{\xi}_\infty\in\mathcal{N}\}
\end{equation}
is well defined.

Let us briefly digress to discuss the ADM energy-momentum and Yang-Mills electric charge. The ADM energy-momentum covector, $\mathbb{P}(g,\pi)=(\breve{E},p_i)$ is usually given by
\begin{align}
16\pi \breve{E}&=\oint_{S_\infty}(\partial_ig_{ij}-\partial_j g_{ii})dS^j\\
16\pi p_i&=2\oint_{S_\infty}\pi_{ij}dS^j,
\end{align}
where the indices refer to some rectangular coordinate system at infinity. We use $\breve{E}$ to indicate the energy, to avoid confusion with the electric field.

We also introduce the standard definition of Yang-Mills electric charge
\begin{equation}
16\pi Q_a=4\oint_{S_\infty}E_{ai}dS^i=-\oint_{S_\infty}\varepsilon_{ai}dS^i,
\end{equation}
which clearly agrees with (\ref{YMcharge}) and the usual Maxwell total electric charge.

It will be more convenient to work with integrals over $\mathcal{M}$ of divergences, rather than surface integrals at infinity. For a fixed $\xi_\infty\in\mathbb{R}^{3,1}\oplus\lie{g}$, we define $\mathbb{P}$ in terms of its pairing with $\xi_\infty$
\begin{align}
16\pi\xi^0_\infty\mathbb{P}_0(g)&=\int_\mathcal{M}\left(\hat{\xi}^0_\infty(\onabla^i\onabla^jg_{ij}-\mathring{\Delta}\text{tr}_{\mathring{g}} g)+\onabla^i \hat{\xi}^0_\infty(\onabla^j g_{ij}-\onabla_i \text{tr}_ {\mathring{g}} g)\right) \sqrt{\mathring{g}}\\
16\pi\xi^i_\infty\mathbb{P}_i(\pi)&=2\int_\mathcal{M}\left(\hat{\xi}^i_\infty\onabla_j\pi_i^j+\pi^{ij}\onabla_i\hat{\xi}_{\infty j}\right)\\
16\pi\xi^a_\infty\mathbb{P}_a(\varepsilon)&=4\int_\mathcal{M}\left(\hat{\xi}^a_\infty\onabla_i E^i_a +E^i_a\onabla_i\hat{\xi}^a_\infty\right),
\end{align}
where indices are raised and lowered using the background metric, $\mathring{g}$. We know from \cite{phasespace} that $\mathbb{P}_\mu=(\mathbb{P}_0,\mathbb{P}_i)$ is $C^\infty$ on the constraint submanifold (if the source is integrable) and the definition is independent of $\mathring{g}$. Since $E$ is a density, the definition of $\mathbb{P}_a$ is clearly independent of $\mathring{g}$ and it is straightforward to check it is smooth.
\begin{proposition}
If $\sigma\in L^1(\Lambda^0(\mathcal{M})\otimes\lie{g})$ and $(s,S_i,\sigma)\in\mathcal{N}^*$, then $\mathbb{P}_a(\varepsilon)$ is a smooth function on the constraint submanifold,
\begin{equation}
\mathbb{P}_a\in C^\infty(\mathcal{C}(s,S,\sigma)).
\end{equation}
\end{proposition}
\begin{proof}
Fix some $\xi_\infty$ and choose a representative $\hat{\xi}_\infty\in\xi_\infty+\mathcal{N}$. Since $\onabla(\hat{\xi}_\infty)$ is compactly supported, we have
\begin{align*}
4\pi|\xi_\infty^a\mathbb{P}_a|&\leq\|\hat{\xi}_\infty\onabla E\|_{1}+\|E \onabla\hat{\xi}_\infty\|_{1}\\
&\leq \|\hat{\xi}_\infty\|_\infty(\|\Phi_a\|_1+\|[A,E]\|_1)+\|E\|_{2,-3/2}\\
&\leq \|\hat{\xi}_\infty\|_\infty(\|\sigma\|_1+\|A_\lie{k}\|_{1,2,-3/2}\|E\|_{1,2,-3/2})+\|E\|_{2,-3/2}.
\end{align*}
$\mathbb{P}_a$ depends linearly on $E$ and is bounded, completing the proof.
\end{proof}

For $\xi\in \xi_\infty+\mathcal{N}$, we define the modified Hamiltonian,
\begin{equation}
\mathcal{H}^{RT}(G;\xi)=16\pi\xi^\alpha\mathbb{P}_\alpha-\int_\mathcal{M}\xi^\alpha\Phi_\alpha,\label{RThamiltonian}
\end{equation}
recalling $G=(g,A,\pi,\varepsilon)$.

On shell, this new Hamiltonian now gives a value for some kind of total energy of the system. One should note that critical points of this Hamiltonian correspond to constrained critical points of the energy, with $\xi$ acting as the (infinite dimensional) Lagrange multiplier. Unfortunately, neither of the terms in (\ref{RThamiltonian}) are well defined on all of $\mathcal{F}$, however, it will be shown that the dominant terms cancel out. To see this, we define the regularised Hamiltonian
\begin{align}\label{reghamil}
\mathcal{H}(G;\xi)&=\int_\mathcal{M}(\hat{\xi}_\infty-\xi)\cdot\Phi+\int_\mathcal{M}\hat{\xi}^0_\infty(\onabla^i\onabla^jg-\mathring\Delta(\tr_{\mathring{g}}g)\sqrt{\mathring{g}}-\Phi_0)\\
&+\int_\mathcal{M}\onabla^i\hat{\xi}^0_\infty(\onabla^j g_{ij}-\onabla_i\tr_{\mathring{g}}g)\sqrt{\mathring{g}}+\int_\mathcal{M}\hat{\xi}^i_\infty(2\onabla_j\pi^j_i-\Phi_i)\nonumber\\
&+\int_\mathcal{M}2\pi^{ij}\onabla_i\hat{\xi}_{\infty j}+\int_\mathcal{M}\hat{\xi}^a_\infty(4\onabla_i E^i_a-\Phi_a)+4\int_\mathcal{M}E^i_a\onabla_i\hat{\xi}^a_\infty.\nonumber
\end{align}
We have combined the terms in (\ref{RThamiltonian}) and then separated them out into 7 integrals, each of which can be shown to be finite.
\begin{theorem}\label{thmhamiltonian}
The regularised Hamiltonian (\ref{reghamil}) is a smooth functional on $\mathcal{F}\times (\xi_\infty+\mathcal{N})$. Furthermore, if $\xi\in \xi_\infty+W^{2,2}_{-1/2}$, then for all $G\in\mathcal{F}$ and $Z\in T_{G}\mathcal{F}$, we have
\begin{equation}
D\mathcal{H}_{(G;\xi)}(Z)=-\int_\mathcal{M}Z\cdot D\Phi_{G}^*(\xi)\label{hamilintegration}.
\end{equation}
\end{theorem}
Equivalently, the regularised Hamiltonian generates the correct equations of motion.
\begin{proof}
First we must establish boundedness and then smoothness follows from the same argument as Corollary \ref{smoothness}. The first integral is easily bounded by $\|\xi-\hat{\xi}_\infty\|_{2,-1/2}\|\Phi\|_{2,-5/2}$. The second and fourth integrals are bounded by Proposition 4.2 of \cite{phasespace}. The 3rd, 5th and 7th integrals are bounded because $\onabla\hat{\xi}_\infty$ has compact support, leaving only the 6th term; for this we note
\begin{equation}
4\onabla\cdot E_a-\Phi_a=4[A,E]_a,
\end{equation}
which is easily taken care of. This establishes the smoothness of $\mathcal{H}$, we now prove the validity of (\ref{hamilintegration}).

Theorem \ref{admhamil} allows us to rewrite the variation of the first term in (\ref{reghamil}) as $\int_\mathcal{M}Z\cdot D\Phi_{G}^*(\hat{\xi}_\infty-\xi)$; we will consider variations of the remaining terms separately. Consider the variation of the  second and third terms,
\begin{align}\label{cancel1}
\int_\mathcal{M}&\Big{\{}\onabla^i(\hat{\xi}^0_\infty(\onabla^j h_{ij}-\onabla_i\tr_{\mathring{g}}h))\sqrt{\mathring{g}}-\nabla^i(\hat{\xi}^0_\infty(\nabla^jh_{ij}-\nabla_i\tr_g h))\sqrt{g}\\
&+\nabla^i(h_{ij}\nabla^j\hat{\xi}^0_\infty-\tr_g h\nabla_i\hat{\xi}^0_\infty)\sqrt{g}-(h,b,p,f)\cdot D\Phi_{(g,A,\pi,\varepsilon)}^*(\hat{\xi}^0_\infty)\Big{\}}\nonumber,
\end{align}
where the middle two terms in the above expression arise from the difference 
\begin{equation*}Z\cdot D\Phi_{G}^*(\hat{\xi}^0_\infty)-\hat{\xi}^0_\infty\cdot D\Phi_{G}(Z).\end{equation*}
The third term in expression (\ref{cancel1}) vanishes as $\onabla(\hat{\xi}^0_\infty)$ has compact support. The dominant terms in the first two divergences cancel, leaving us with a boundary term of the form $\hat{\xi}^0_\infty\tilde{\Gamma}h$, ignoring factors of $g$. Note $\hat{\xi}^0_\infty\in L^\infty$, leaving us a boundary term $\tilde{\Gamma}h$, which is of the exact form of $\mathcal{B}^1$ considered above and thus contributes nothing. Similarly, the variation of the fourth and fifth terms in (\ref{reghamil}) give
\begin{equation}\label{cancel2}
2\int_\mathcal{M}\left\{\onabla_i(\hat{\xi}^k \mathring{g}_{jk}p^{ik})-\nabla_i(\hat{\xi}^j_\infty p^i_j)-Z\cdot D\Phi_{i (G)}^*(\hat{\xi}^i_\infty)\right\}.
\end{equation}
The first and second terms in (\ref{cancel2}) give the boundary term $(g-\mathring{g})p$, which is again of the form of $\mathcal{B}^1$ above. Finally, the variation of the 6th and 7th terms in (\ref{reghamil}) give
\begin{equation}
\int_\mathcal{M}\left\{-4\onabla_i(\hat\xi^a_\infty f^i_a)+4\nabla_i(\hat{\xi}^a_\infty f_a^i) -Z\cdot D\Phi_{a(G)}^*(\hat{\xi}^a_\infty)\right\},
\end{equation}
where the first and second terms here exactly cancel. Putting all of this together completes the proof
\end{proof}

\section{The First Law}\label{Sfirstlaw}
It is well known that there is a strong analogy between the laws of thermodynamics and those of black holes. The first law is usually expressed by the in terms of differentials as,
\begin{equation}
dm=\frac{\kappa}{8\pi}d\hat{A}+\Omega dJ+VdQ\label{firstlaw},
\end{equation}
valid for perturbations of stationary solutions. Here, the quantities $m$, $\kappa$, $\hat{A}$, $\Omega$, $J$, $V$ and $Q$ correspond to the mass, surface gravity, horizon area, angular velocity, angular momentum, electric potential and electric charge of the black hole, respectively. With our conditions for $\mathcal{M}$ and $g$, there will be no black hole present so this expression will reduce significantly. Also, interpreting $V$ as the potential difference between the horizon and infinity, we expect to replace $V$ with $-V_\infty$ in this expression, leading us to
\begin{equation}\label{firstlaw2}
dm+V_\infty dQ=0.
\end{equation}
Theorem \ref{main} and the subsequent corollary provide a proof of this similar to that of Sudarsky and Wald, as well as a converse statement conjectured in \cite{SW1}, which could not be rigorously shown without the Banach manifold structure given in Section \ref{Sconstraint}. We prove that a solution satisfying this version of the first law must be stationary. Sudarsky and Wald also conjectured that a similar result should hold when an interior boundary is present, however it is likely that the precise form of (\ref{firstlaw}) be modified or the phase space include boundary conditions. A potential candidate for suitable boundary conditions are the conditions of an isolated horizon; as Ashtekar, Fairhurst, Krishnan and Beetle (see \cite{AFKfirstlaw, Beetlethesis} and references therein) have established a local version of the first law for isolated horizons. However this is beyond the scope of this paper and will be considered in future work.

To begin, let us first quote the following generalisation of the method of Lagrange multipliers to Banach manifolds (Theorem 6.3 of \cite{phasespace}).
\begin{theorem}\label{banach}
Suppose $K:B_1\rightarrow B_2$ is a $C^1$ map between Banach manifolds, such that the map, ${DK_u:T_uB_1\rightarrow T_{K(u)}B_2}$, is surjective with closed kernel and closed complementary subspace for all $u\in K^{-1}(0)$. Let $f\in C^1(B_1)$ and fix $u\in K^{-1}(0)$; the following statements are equivalent:
\begin{enumerate}[(i)]
\item For all $v\in\ker DK_u$, we have
\begin{equation}Df_u(v)=0.\end{equation}
\item There is $\lambda\in B_2^*$ such that for all $v\in B_1$,
\begin{equation}Df_u(v)=\left<\lambda,DK_u(v)\right>,\end{equation}
where $\left< \, , \right>$ refers to the natural dual pairing.
\end{enumerate}
\end{theorem}
\vspace{6mm}
We can now apply this to prove the main result.
\begin{theorem}\label{main}
Let $G=(g,A,\pi,\varepsilon)\in\mathcal{F}$ be such that $\Phi(G)=(s,S_i,\sigma_a)$, for some \\${(s,S_i,\sigma_a)\in L^1\left( \Lambda^3\times T^*\mathcal{M}\otimes\Lambda^3\times \mathfrak{g}^*\otimes\Lambda^3\right)}$. Further, let $(\xi_\infty^\mu,\xi_\infty^a)\in\mathbb{R}^{3,1}\oplus\mathfrak{g}$ be fixed and define the energy functional $E\in C^\infty(\mathcal{C}(s,S,\sigma))$ by
\begin{equation}
E(G)=\xi^\alpha_\infty\mathbb{P}_\alpha(G),
\end{equation}
The following statements are equivalent:
\begin{enumerate}[(i)]
\item For all $Z=(h,b,p,f)\in T_G\mathcal{C}(s,S,\sigma)$
\begin{equation}
DE_G (Z)=0.
\end{equation}
\item There is $\xi\in\xi_\infty+W^{2,2}_{-1/2}\left( (\mathcal{T}\times\mathfrak{g})\right)$ satisfying
\begin{equation}
D\Phi_G^*\xi=0.
\end{equation}
\end{enumerate}
\end{theorem}
\begin{proof}
First we show $(i)\Rightarrow (ii)$. For any fixed $\tilde{\xi}\in\xi_\infty+W^{2,2}_{-1/2}$, define $f(G')=\mathcal{H}(G';\tilde{\xi})$ for all $G'\in\mathcal{F}$ and $K=\Phi-(s,S,\sigma)$. With $u=G$, we now have the hypotheses of Theorem \ref{banach}. Note also, $T_G\mathcal{C}(s,S,\sigma)=\ker(DK_G)$ and on $T_G\mathcal{C}(s,S,\sigma)$, we have $Df_G=16\pi DE_G$. So directly applying Theorem \ref{banach}, there exists $\lambda\in\mathcal{N}$ such that for all $Z\in T_G\mathcal{F}$,
\begin{equation}
Df_g=\int_\mathcal{M}\lambda\cdot D\Phi^*_G(Z).
\end{equation}
Inserting the definition of $f$, we have
\begin{equation}
D\Phi^*_G(\lambda)=-D\Phi^*_G(\tilde{\xi})\hspace{20mm} \text{(weakly)}.
\end{equation}
Applying Theorem \ref{weakstrong} we have
\begin{equation}
D\Phi^*_G(\xi)=0,
\end{equation}
where $\xi=\lambda+\tilde{\xi}\in\xi_\infty+W^{2,2}_{-1/2}$. We now have $(i)\Rightarrow (ii)$.\\
To show $(ii)\Rightarrow (i)$, we simply must recall that $16\pi DE_G(Z)=D\mathcal{H}_G(Z)$ for all $Z\in T_G\mathcal{C}(s,S,\sigma)$ then from Theorem \ref{thmhamiltonian} we have $16\pi DE_G(Z)=D\mathcal{H}_G(Z)=-\int_\mathcal{M}Z\cdot D\Phi_G^*(\xi)=0$\\
\end{proof}
Recall, a solution $(N^\mu,V^a)=\xi$ of $D\Phi_G^*(\xi)=0$ corresponds to a generalised Killing vector $N^\mu$ and `electric potential' $V^a$, representing an infinitesimal symmetry in the bundle. That is, evolution along integral curves of $\xi$ in the bundle leaves the data fixed (see \cite{armsEYM}). It was shown by Beig and Chru\'sciel\cite{BeigChruscielKilling}, that if a Killing vector is timelike at infinity then it is asymptotically proportional to $\mathbb{P}^\mu=\eta^{\mu\nu}\mathbb{P}_{\nu}$, where $\eta$ is the Minkowski metric. It was further shown that, provided $T_{00}^2\geq T_{i0}T^{i0}$ and $\mathbb{P}\neq0$, $\zeta^\mu\mathbb{P}_\mu> 0$ for all future timelike vectors, $\zeta$. We will say the covector $\mathbb{P}_\mu$ is future timelike if $\eta^{\mu\nu}\mathbb{P}_\nu$ is timelike and $\zeta^\mu\mathbb{P}_\mu> 0$ for all future timelike vectors, $\zeta$.
\begin{corollary}\label{thecorollary}
Suppose $G\in\mathcal{F}$, $\Phi(G)=(s,S,\sigma)\in L^1$ and $\mathbb{P}_\mu$ is future timelike, then the following statements are equivalent:
\begin{enumerate}[(i)]
\item
For all $Z\in T_G\mathcal{C}(s,S,\sigma)$;
\begin{equation}
Dm_G(Z)+V_\infty\cdot DQ_G(Z)=0\label{result},
\end{equation}
where $m=\sqrt{-\mathbb{P}^\mu\mathbb{P}_\mu}$ is the total (or ``rest") ADM mass, $V_\infty\in\lie{g}$ is the Yang-Mills electric potential at infinity, and $Q_a=\frac{1}{4\pi}\oint_\infty E_{ai}ds^i$ is the Yang-Mills electric charge.
\item
$G$ is a generalised stationary initial data set with infinitesimal symmetry generator $(N^\mu,V^a)=\xi$, in the sense $D\Phi_G^*(\xi)=0$, and $N^\mu_\infty$ is proportional to $\mathbb{P}^\mu$.
\end{enumerate}
\end{corollary}
\begin{proof}
First we show $(i)\Rightarrow (ii)$. Choose $N^\mu_\infty=-\frac{1}{m}\eta^{\mu\nu}\mathbb{P}_\nu$, a future pointing unit timelike vector, such that $N^\mu_\infty\mathbb{P}_\mu=m$. For $\xi_\infty=(N^\mu_\infty,V^a_\infty)$, we have $DE=\xi^\alpha D\mathbb{P}_\alpha=Dm+V\cdot DQ$ and thus Theorem \ref{main} gives us $(ii)$.

Conversely, $(ii)$ implies the condition $(ii)$ of Theorem \ref{main}, so we have $DE_G(Z;N_\infty,V_\infty)=0$ for all $Z\in T_G\mathcal{C}(s,S,\sigma)$. Since we have $N^\mu_\infty$ proportional to $\mathbb{P}^\mu$, we can rescale $\xi$ such that $N^\mu_\infty\mathbb{P}_\mu=m$ again, which completes the proof.
\end{proof}
\section{Acknowledgements}
The author would like to thank Robert Bartnik for guidance and advice, as well as the School of Mathematical Sciences at Monash University for their support and hospitality.
\nocite{beigomurchadha,GT}
\bibliographystyle{plain}
\bibliography{../../refs}
\end{document}